\documentclass[runningheads]{llncs}
\usepackage{amsmath}
\usepackage{amssymb}
\usepackage{mathtools}
\usepackage{stmaryrd}
\usepackage{dsfont}
\usepackage{tikz}

\newcommand{\N}{\mathbb{N}} 
\renewcommand{\phi}{\varphi}
\renewcommand{\implies}{\supset}

\newcommand{\FO}{\mathsf{FO}}
\newcommand{\SFO}{2\mathsf{Sor.FO}}
\newcommand{\FOML}{\mathsf{FOML}}

\newcommand{\TML}{\mathsf{TML}}
\newcommand{\PTML}{\mathsf{PTML}}
\newcommand{\IQML}{\mathsf{IQML}}

\newcommand{\existsBox}{[\exists]}
\newcommand{\forallBox}{[\forall]}
\newcommand{\existsDiamond}{\langle \exists \rangle}
\newcommand{\forallDiamond}{\langle \forall \rangle}

\newcommand{\veeS}{\overset{\vee}{S}}

\newcommand{\Ps}{\mathcal{P}}
\newcommand{\Var}{\mathcal{V}}
\newcommand{\VarX}{\mathcal{V}_X}
\newcommand{\VarTau}{\mathcal{V_\tau}}

\newcommand{\Fv}{\mathsf{Fv}}
\newcommand{\md}{\mathsf{md}}

\newcommand{\Tr}{\mathsf{Tr}}

\newcommand{\M}{\mathcal{M}}
\newcommand{\W}{\mathcal{W}}
\newcommand{\D}{\mathcal{D}}
\newcommand{\R}{\mathcal{R}}
\newcommand{\Rs}{\mathds{R}}
\newcommand{\I}{\mathds{I}}

\newcommand{\Mfo}{\mathfrak{M}}

\newcommand{\EF}{\mathsf{EF}}
\newcommand{\PI}{\mathbf{Sp}}
\newcommand{\PII}{\mathbf{Dup}}

\newcommand{\AxA}{\mathcal{AX_A}}

\newcommand{\bisEq}{\leftrightarroweq}
\newcommand{\elemEq}{\equiv_\IQML}

\newcommand{\PSPACE}{\mathsf{PSPACE}}

\newcommand{\BR}{\mathsf{BR}}

\begin{document}
\title{Propositional modal logic with implicit modal quantification}
%
%
\author{Anantha Padmanabha\inst{1}\orcidID{0000-0002-4265-5772} \and
R Ramanujam\inst{1}
\authorrunning{Padmanabha and Ramanujam}}
%
\institute{Institute of Mathematical Sciences, HBNI, Chennai, India 
\email{\{ananthap,jam\}@imsc.res.in}}
\maketitle              
\begin{abstract}
Propositional term modal logic is interpreted over Kripke structures with
unboundedly many accessibility relations and hence the syntax admits variables
indexing modalities and quantification over them. This logic is undecidable, 
and we consider a variable-free propositional bi-modal logic with implicit
quantification. Thus $\forallBox \alpha$ asserts necessity over all 
accessibility relations and $\existsBox \alpha$ is classical necessity over
some accessibility relation. The logic is associated with a natural 
bisimulation relation over models and we show that the logic is exactly
the bisimulation invariant fragment of a two sorted first order logic.
The logic is easily seen to be decidable and admits a complete axiomatization
of valid formulas. Moreover the decision procedure extends naturally to the
`bundled fragment' of full term modal logic.

\end{abstract}
\section{Introduction}
Propositional multi-modal logics \cite{blueBook,hughesBook} are used extensively in
the context of multi-agent systems, or to reason about labelled transition systems.
In the former case, $\Box_i \alpha$ might refer to knowledge or belief of agent $i$
that $\alpha$ holds. In the latter case, $\Diamond_a \alpha$ may assert the existence
of an $a$-labelled transition from the current state to one in which $\alpha$ holds.
Such applications include  epistemic reasoning \cite{epistemicBook,hansBook}, 
games\cite{van2007}, system verification\cite{temporalPaper,atlPaper} and more.

In either of the settings, the indices of modalities come from a fixed finite set.
However, the applications themselves admit systems of unboundedly many agents,
or infinite alphabets of actions. The former is the case in dynamic networks of
processes, and the latter in the case of systems handling unbounded data. In fact,
the set of agents relevant for consideration may itself be dynamic, changing
with state. 

Such  motivations naturally lead to modal logics with unboundedly
many modalities, and indeed quantification over modal indices. Grove and Halpern 
\cite{grove1,grove2} discuss epistemic logics where the agent set is not fixed 
and the agent names are not common knowledge. Khan et al. \cite{aquil} use 
unboundedly many modalities and allow quantification over them to model 
information systems in approximation spaces. Other works on indexed 
modalities include Passy and Tinchev \cite{Passy}, Gargov and Goranko 
\cite{Gargov}, Blackburn \cite{Blackburn}.

Term Modal logic($\TML$), introduced by Fitting, Voronkov and Thalmann\cite{TML} 
offers a natural solution to these requirements. It extends first order logic
with modalities of the form $\Box_x \alpha$ where $x$ is a variable (and hence
can be quantified over).  Thus we can write a formula of the form: 
$\forall x  \Box_x(  p(x) \implies \exists y \Diamond_y q(x,y))$.   
Kooi (\cite{kooi2007}) considers the expressivity of $\TML$ in epistemic 
setting.  Wang and Seligman (\cite{WangTML}) introduce a restricted version 
of $\TML$ where  we have assignments in place of quantifiers (formulas of 
the form $[x:= b]K_x(\alpha)$ where $b$ is a constant, whose 
interpretation as an agent will be assigned to $x$). 
    
Note that $\TML$ extends first order logic, and hence its satisfiability problem 
is undecidable. In \cite{PR17} we prove that the problem  is undecidable even 
when the atoms are restricted to boolean propositions ($\PTML$). Hence the question
of finding decidable fragments of $\PTML$ is well motivated. 
In \cite{PR17} we prove that  the {\em monodic fragment} of $\PTML$ is
decidable. The monodic fragment is a restriction allowing at most one 
free variable within the scope of a modality. i.e, every subformula of 
the form $\Box_x \alpha$ has $FV(\alpha) \subseteq \{x\}$. 

Orlandelli and Corsi \cite{orlandelli} consider two decidable fragments: $(1)$ When  
quantifier  occurrence is restricted to the form:  $\exists x \Box_x \alpha$ 
(denoted by $\existsBox \alpha$); $(2)$ Quantifiers appear in a restricted guarded form: 
$\forall x (P(x) \Rightarrow \Box_x \alpha)$ and $\exists x (P(x) \land 
\Box_x \alpha)$ (and their duals).  The corresponding first order modal logic counter parts of the first of these fragments is studied by Wang \cite{Wang17}. 
Shtakser (\cite{shtakser2018}) considers a monadic second order version of the restricted guards (with propositional atoms) of the form $\forall X (P(X) \Rightarrow \Box_X \alpha)$ 
and $\exists X (P(X) \land \Box_X \alpha)$ where $X$ is quantified over subsets of 
indices and $P$ is interpreted appropriately.  
   These fragments are semantically motivated from their interest in the epistemic logic to model the notions like `everyone knows' and `someone knows' and community knowledge (ex: All eye-witnesses know who killed Mary). 
    
Note that when modalities and quantifiers are `bundled' together and atomic formulas
are propositional, $\exists x \Box_x \alpha$ can be replaced by a variable free
modality $\existsBox \alpha$, and similarly $\forall x \Box_x \alpha$ by
$\forallBox$. In some sense this is the most natural variable free fragment of
$\PTML$ with modalities being {\em implicitly quantified}. This is the logic
$\IQML$ studied in this paper.

Just as propositional modal logic is the bisimulation-invariant fragment of 
first order logic, we show that $\IQML$ is the bisimulation-invariant fragment of
an appropriate two-sorted first order logic. The notion of bisimulation needs
to be carefully re-defined to account for quantification over edge labels. Other
natural questions on $\IQML$ such as decidability of satisfiability and complete
axiomatization of valid formulas are answered easily. Interestingly, the natural
tableau procedure for the logic can be extended to the `bundled fragment' of 
$\TML$ with predicates of arbitrary arity, by an argument similar to the one
developed in \cite{PRW18} (for a `bundled fragment' of first order modal logic).

\section{The logic}
 \label{sec: syntax and semantics}
 
 We start with $\PTML$, the propositional fragment of Term-Modal logic. Since we
will only study its variable free fragment later, we consider here only the 
pure vocabulary (no constant and function symbols) with only variables as terms 
and without equality.

\begin{definition}[$\PTML$ syntax]
\label{def: PTML syntax}
Let $\Var$ be a countable set of variables and $\Ps$ be a countable set of propositions. The syntax of $\PTML$ is given by:

$$\phi := p \mid \neg \phi \mid \phi \land \phi \mid \exists x\ \phi \mid \Diamond_x \phi$$ where $p \in \Ps$ and $x\in \Var$.
\end{definition}

The boolean operators $\lor,  \implies$ are defined in the standard way. The dual operators for quantifiers and modalities are given by $ \forall x \phi = \neg \exists x\ \neg \phi$ and $\Box_x \phi = \neg \Diamond_x \neg \phi$. The notion of free variables $\Fv(\phi)$ and modal depth $\md(\phi)$ are standard.

In the semantics, unlike classical modal logics, the agent set is not fixed, but specified along with the structure. Thus the Kripke frame for $\PTML$ is given by $(\W,\D,\R)$ where $\W$ is a set of worlds, $\D$ is a potential set of agents and $\R\subseteq (\W\times \D \times \W)$. The agent dynamics is captured by a function
($\delta: \W \rightarrow 2^\D$ below) that specifies, at any world $w$, the set of agents {\em alive}
(or meaningful) at $w$. Then coherence demands that whenever $(u,d,v)\in \R$, we
have that $d \in \delta(u)$: only an agent alive at $u$ can consider $v$ accessible.

A {\em monotonicity} condition  is imposed on the accessibility relation as well:
whenever $(u,d,v)\in \R$, we have that $\delta(u)\subseteq \delta(v)$. This
is required to handle interpretations of free variables.  Hence the models are called `increasing agents' models. For more details on this restriction,  refer \cite{FOML,TML}.
 
\begin{definition}[$\PTML$ structure]
\label{def: TML structure}
An (increasing agent) model for $\PTML$ is a tuple $\M=(\W, \D, \R, \delta, \rho)$ where, 
$\W$ is a non-empty set of worlds, $\D$ is a non-empty set of agents, $\R\subseteq (\W\times \D\times \W)$,
$\delta:\W\to 2^\D$ assigns to each $w\in \W$ a \textit{non-empty} local agent set s.t. 
$(w,d,v) \in \R$ implies $d\in \delta(w)\subseteq \delta(v)$ for any $w,v\in \W$, and 
$\rho: \W \mapsto 2^\Ps$.
\end{definition}

To interpret free 
variables, we  need a variable assignment function (interpretation) $\sigma: \Var\mapsto \D$.  Call $\sigma$
{\em relevant} at $w \in \W$ for a formula $\phi$ if $\sigma(x)\in \delta(w)$ for all $x\in \Fv(\phi)$. The
increasing agent condition ensures that whenever $\sigma$ is relevant at $w$ for $\phi$ and
we have $(w,d,v)\in \R$, then $\sigma$ is relevant at $v$ for all subformulas of $\phi$. 

\begin{definition}[$\PTML$ semantics]
\label{def: PTML semantics}
Given a model $\M$, a formula $\phi$,\ $w \in \W^\M$, and an interpretation $\sigma$ that is relevant at $w$ for $\phi$, define $\M,w,\sigma \models \phi$ inductively as follows:
$$\begin{array}{|lcl|}
\hline
\M, w, \sigma\models p&\Leftrightarrow &  p\in \rho(w) \\ 
\M, w, \sigma\models \neg\phi &\Leftrightarrow&   \M, w, \sigma\not\models \phi \\ 
\M, w, \sigma\models (\phi\land \psi) &\Leftrightarrow&  \M, w, \sigma\models \phi \text{ and } \M, w, \sigma\models \psi \\ 
M, w, \sigma\models \exists x \phi &\Leftrightarrow& \text{there is some $d\in \delta(w)$ such that } \M, w, \sigma[x\mapsto d]\models\phi \\
\M, w, \sigma\models \Diamond_x \phi &\Leftrightarrow& \text{there is some } v\in W^\M \text{ such that } (w,\sigma(x),v) \in \R^\M \\ &&\text{ and }  \M, v, \sigma\models\phi\\

\hline
\end{array}$$

\noindent where $\sigma[x\mapsto d]$ denotes an interpretation that is the same as $\sigma$ 
except for mapping $x$ to $d$. 
\end{definition}

 Note that $\M,w,\sigma \models \phi$ is inductively defined only when $\sigma$ is 
relevant at $w$. A formula $\phi$ is \textit{satisfiable}, if there is some $\M$ and some $w\in \W^\M$ and an interpretation $\sigma$ which is relevant at $w$ for $\phi$ such that $\M,w,\sigma \models \phi$. Also, $\phi$ is \textit{valid} if $\neg \phi$ is not satisfiable. In \cite{PR17}, we prove that the satisfiability problem for $\PTML$ is undecidable.

As discussed in the previous section, we consider the variable free fragment of $\PTML$,
with implicit modal quantification ($\IQML$).

\begin{definition}[$\IQML$ syntax]
\label{def: IQML syntax}
Let $\Ps$ be a countable set of propositions. The syntax of $\IQML$ is given by:
$$\phi := p \in \Ps \mid \neg\phi \mid \phi \land \phi \mid \existsBox \phi \mid \forallBox \phi$$
\end{definition}

Note that,  $\existsBox \phi$ translates to $\exists x \Box_x \phi$ in $\PTML$. Similarly $\forallBox \phi$ translates to $\forall x \Box_x \phi$. Since there are no variables in $\IQML$, it is closer to classical propositional modal logics where the set of modal indices is not fixed a priori.

The boolean operators $\vee$ and $\implies$ are defined in the standard way. Also we define $\forallDiamond \phi = \neg \existsBox \neg \phi$ and $\existsDiamond \phi = \neg \forallBox \neg \phi$ to be the respective duals of the modal operators.

 In classical modal logics, the  Kripke structure for $n$ modalities is given by $\M = (\W, \R_1,\cdots \R_n, \rho)$ where each $R_i \subseteq (\W \times \W)$ is the accessibility relation for the corresponding index and $\rho$ is the valuation of propositions at every world. But in case of $\IQML$, the modal index set is specified along with the model.

\begin{definition}[$\IQML$ structure]
\label{def: IQML structure}
An $\IQML$ structure is given by $\M = (\W,\Rs_\I,\rho)$ where $\W$ is a non-empty set of worlds, $\I$ is a non-empty countable index set and $\Rs = \{ \R_i \mid i\in \I\}$ where each $\R_i \subseteq(\W \times \W)$ and $\rho: \W \mapsto 2^\Ps$ is the valuation function.
\end{definition}

Note that $\I$ could be finite or countably infinite. Hence we assume $\I$ to be some 
initial segment of $\N$ or $\N$ itself. Thus we often denote the model as 
$\M = (\W, [\R_1, \R_2 \cdots],\rho)$ when $\I$ is clear from the context. Given a model $\M$, we refer to $\W^\M$ etc. to denote its corresponding components.
The semantics is defined naturally as follows:

\begin{definition}[$\IQML$ semantics]
\label{def: IQML semantics}
Given a model $\M$, a formula $\phi$,\ $w \in \W^\M$, define $\M,w \models \phi$ inductively as follows:
$$\begin{array}{|lcl|}
\hline
\M, w\models p&\Leftrightarrow &  p\in \rho(w) \\ 
\M, w\models \neg\phi &\Leftrightarrow&   \M, w, \not\models \phi \\ 
\M, w\models (\phi\land \psi) &\Leftrightarrow&  \M, w \vDash \phi \text{ and } \M, w \models \psi \\ 
M, w\models \existsBox \phi &\Leftrightarrow& \text{there is some } i\in \I \text{ such that for all } u\in \W\\
&& \text{ if } (w,u) \in R_i \text { then } \M, u\models\phi \\
\M, w \vDash \forallBox \phi &\Leftrightarrow& \text{for all } i\in \I \text{ and for all } u\in \W\\
&& \text{ if } (w,u) \in R_i \text { then } \M, u\models\phi \\

\hline
\end{array}$$
\end{definition}

The formula $\phi \in \IQML$ is \emph{satisfiable} if there is some model $\M$ and $w\in \W$ such that $\M,w\models \phi$. A formula $\phi$ is said to be \emph{valid} if $\neg \phi$ is not satisfiable. 

In the sequel we adopt the following convention. Given any model $\M$, $w\in \W$ and  
a formula of the form $\existsBox \phi$, if $\M,w \models \existsBox \phi$ and 
$i\in \I$ is the corresponding witness then we write $\M,w \models \Box_i \phi$ 
(similarly we have $\M,w \models \Diamond_i \phi$ for $\existsDiamond \phi$). 

\section{Axiom system and completeness}

Table \ref{tab: IQML axioms} gives a complete axiom system for  the valid formulas of $\IQML$. 

The axioms and inference rules are standard.  Axiom $A2$ describes the interaction 
between $\forallBox$ and $\forallDiamond$ operators.  The ($\existsBox$Nec) rule
is  sound since $\I$ is non-empty.  Note that the axiom system is similar to 
the one in \cite{grove1}, except for ($\forallBox$Nec) and ($\existsBox$Nec). This 
is because $\IQML$ has no names, as opposed to the logic considered in \cite{grove1}.

\begin{table}
\label{tab: IQML axioms}
\centering
\begin{tabular}{| l l |}
\hline
&$\vdash_{\AxA}$\\
\hline
\label{propositional packed axioms}
$A0.$& All instances of propositional validities.\\
$A1.$ & $\forallBox (\phi \implies \psi) \implies (\forallBox \phi 
\implies \forallBox \psi)$\\

$A2.$ & $\forallBox (\phi \implies \psi) \implies ( \forallDiamond \phi \implies \forallDiamond \psi)$\\

(MP)& \begin{tabular}{c}
		$\phi\implies\psi , \phi$\\
		\hline
		$\psi$
		\end{tabular} \\
		
($\forallBox$Nec)& \begin{tabular}{c}
		$\phi$\\
		\hline
		$\forallBox \phi$
		\end{tabular}\\
		
($\existsBox$Nec)& \begin{tabular}{c}
		$\phi$\\
		\hline
		$\existsBox \phi$
		\end{tabular}\\

\hline
\end{tabular}
\caption{$\IQML$ axiom system $(\AxA)$}
\end{table}

\begin{theorem}
\label{thm: IQML completeness}
$\vdash_{\AxA}$ is sound and complete for $\IQML$.
\end{theorem}

We shall first prove soundness.

\begin{lemma}
\label{soundness}
The axiom system $\vdash_{\AxA}$ is sound for $\IQML$.
\end{lemma}
\begin{proof}
To see that $A2$ is a validity, for any model $\M$ and any world $w$ let $\M,w \models \forallBox (\phi\implies\psi)$ and $\M,w \models \forallDiamond \phi$. Since $\M,w\models \forallBox (\phi\implies\psi)$ for any $i\in \I$ and for any $w\xrightarrow{i}u$ we have $\M,u \models \phi \implies \psi$. Further since $\M,w \models \forallDiamond \phi$, for any $i\in \I$ there is some $v$ such that $w\xrightarrow{i}v$ and $\M,v \models \phi$. But then $\M,v \models \phi \implies \psi$ and hence $\M,v \models \psi$. Thus by semantics, $\M,w \models \forallDiamond \psi$.

Similarly validity of $A1$ which is the variant of standard $K$ axiom can be verified. Also notice that the inference rules $MP$ and both ($Nec$)  preserve validities. Hence $\vdash_{\AxA}$ is sound.
\end{proof}

For completeness, we first prove some useful lemmas. The notions of \emph{consistent set of formulas} and \emph{ maximally consistent set of formulas} is defined in the standard way.

\begin{lemma}
\label{e-diamond consistency}
For any set of formulas $\Gamma$, if $\Gamma$ is a maximal consistent set then
\begin{enumerate}
\item if $\existsDiamond \beta \in \Gamma$ then $\{ \beta\} \cup \{ \psi \mid \forallBox \psi\in \Gamma\}$ is consistent.
\item if $\{\forallDiamond \gamma,\ \existsBox \delta \} \subseteq \Gamma$ then $\{ \gamma, \delta\} \cup \{ \psi \mid \forallBox \psi\in \Gamma\}$ is consistent.
\end{enumerate}
\end{lemma}

\begin{proof}
To prove ($1$),  let $\Gamma$ be a maximal consistent set of formulas and $\existsDiamond \beta \in \Gamma$. Define $\Lambda = \{ \beta	\} \cup \{ \psi \mid \forallBox \psi \in \Gamma \}$. We need to prove that $\Lambda$ is consistent. Suppose not, then there are some $\psi_1, \psi_2\cdots \psi_n \in \Lambda$ such that

\begin{tabular}{l l}
 &$\vdash_{\AxA}(\psi_1\land \psi_2 \cdots \psi_n) \implies \neg \beta$.\\
 By ($\forallBox$Nec) we have & $\vdash_{\AxA} \forallBox \big((\psi_1\land \psi_2 \cdots \psi_n) \implies \neg \beta\big)$.\\
By $(A1)$ and (MP),& $\vdash_{\AxA} \forallBox (\psi_1\land \psi_2 \cdots \psi_n ) \implies \forallBox \neg \beta$.
 \end{tabular}
 
\medskip
  Also note that $(\forallBox \psi_1 \land \forallBox \psi_2 \cdots\forallBox \psi_n) \implies \forallBox (\psi_1\land \psi_2\cdots \psi_n)$ is a theorem in this system. Hence $\vdash_{\AxA} (\forallBox \psi_1 \land \forallBox \psi_2 \cdots\implies \psi_n) \implies \forallBox \neg \beta$.
  This implies $\forallBox \neg \beta \in \Gamma$ which is a contradiction to $\existsDiamond \beta \in \Gamma$ and $\Gamma$ is maximally consistent.

  \medskip
To prove ($2$), again let $\Gamma$ be a maximal consistent set of formulas and let   $\{\forallDiamond \gamma,\ \existsBox \delta\} \subseteq \Gamma$. Define $\Lambda = \{ \gamma,\ \delta	\} \cup \{ \psi \mid \forallBox \psi \in \Gamma \}$. We need to prove that $\Lambda$ is consistent. Suppose not, then there are some $\psi_1, \psi_2\cdots \psi_n \in \Lambda$ such that
$\vdash_{\AxA}(\psi_1\land \psi_2 \cdots \psi_n) \Rightarrow (\gamma\Rightarrow \neg \delta)$.

Now arguing in the same way as in $(1)$ we have\\ 
\begin{tabular}{l l}
 &$\Gamma\vdash_{\AxA} \forallBox ( \gamma \Rightarrow \neg \delta)$\\
 By $(A2)$ & $\Gamma \vdash_{\AxA} \forallBox ( \gamma \Rightarrow \neg \delta) 
 \Rightarrow (\forallDiamond \gamma \Rightarrow \forallDiamond \neg \delta)$\\
 By $(MP)$ &$\Gamma \vdash_{\AxA} \forallDiamond \gamma \Rightarrow \forallDiamond \neg \delta$\\
 Since $\forallDiamond\beta \in \Gamma$,& $\Gamma \vdash_{\AxA} \forallDiamond \neg \delta$.
 \end{tabular}
 

  This is a contradiction since $\existsBox \delta \in \Gamma$ and $\Gamma$ is consistent.
\end{proof}

Now we define the canonical model. Let $EB = \{ \existsBox \alpha \mid \existsBox \alpha \in \IQML\}$ be the set of all $\existsBox$ formulas. These formulas will be used as ``agents" in the canonical model.

\begin{definition}
\label{propositional packed cannonical model}
The canonical model for the propositional packed model is given by $\hat{\M} = (\hat{\W}, \hat{\R}_{\hat{\I}}, \hat{\rho})$ where
\begin{itemize}
\item $\hat{\W}$ is set of all \emph{maximal consistent sets}
\item $\hat{\I} =  \{ i_\alpha \mid \alpha \in EB \} \cup \{ j\}$ where $j$ is distinct form all $i_\alpha$.
\item To define $\hat{\R}$, for all $w,u \in \hat{\W}$ and for all $i_{\existsBox\alpha} \in \I$   we have $(w,u) \in \R_{i_{\existsBox\alpha}}$ if $\existsBox \alpha \in w$ and $\{ \alpha \} \cup \{ \psi \mid \forallBox \psi \in w\} \subseteq u$.	

For $j\in \hat{\I}$ we have	$w \xrightarrow{j} u$ if $\{\psi \mid \forallBox \psi \in w\} \subseteq u$.
\item $\hat{\rho}(w) = w \cap \Ps$.
\end{itemize}
\end{definition}

\begin{lemma}
\label{exists witness}
In the canonical model, for any $w,u \in \hat{\W}$ and $i \in \hat{\I}$ if $w \xrightarrow{i} u \in \hat{\R}$ then for all $\psi \in u,\ \existsDiamond \psi \in w$.
\end{lemma}

\begin{proof}
Suppose not, then there is some $w\xrightarrow{i}u$ and some $\psi\in \IQML$ such that $\psi\in u$ but $\existsDiamond \psi \not\in w$. But $w$ is maximal, so $\forallBox \neg \psi \in w$. Now since $w\xrightarrow{i}u$, by definition of $\hat{\R},\ \neg \psi\in u$ which is a contradiction.
\end{proof}

Now we are ready to prove that $\vdash_{\AxA}$ is complete for valid formulas of $\IQML$.

\begin{proof}
We show this by proving that any consistent formula $\phi\in \IQML$ is satisfiable. First note that any consistent set of formulas $\Gamma$ can be extended to a maximal consistent set by the standard Lindenbaum construction.

Hence for any consistent set of formulas $\Gamma$, there is some world $w \in \hat{\W}$  such that $\Gamma \subseteq w$. Now, we prove the truth lemma.

\begin{claim} For any $w\in \hat{\W},\ \hat{\M},w \models \phi$ iff $\phi \in w$.
\end{claim}

The proof is by induction on the structure of $\phi$. In the base case we have propositions and the claim follows by definition of $\hat{\rho}$. The $\neg$ and $\land$ cases are standard.

For the case $\phi := \existsDiamond \beta$, 
suppose $\hat{\M}, w \models \existsDiamond \beta$ then there is some $a \in \hat\gamma(w)$ and some $(w,a,u)\in \hat{\R}$ such that $\hat{\M},u \models \beta$. By induction hypothesis $\beta \in u$ and by lemma \ref{exists witness}, $\existsDiamond \beta\in w$.

For the other direction, suppose $\existsDiamond \beta \in w$ then since $w$ is a consistent set (by lemma \ref{e-diamond consistency}(1)) we have 
$\Gamma = \{ \beta\} \cup \{ \psi \mid  \forallBox \psi \in w\}$ is consistent. Thus there is some world $ u \supseteq \Gamma$. Now since $\beta \in u$, by induction hypothesis $\hat{M}, u \models \beta$ and also since $\{ \psi \mid \forallBox \psi \in w\} \subseteq u$ we have $w \xrightarrow{j}u$ and hence 
$\hat{M}, w \models \existsDiamond \beta$.

\medskip

For the case $\phi := \existsBox \beta$, To prove $(\Rightarrow)$, we consider the contrapositive. We prove that  if $\existsBox \beta \not \in w$ then $\hat{M},w \models \forallDiamond \neg \beta$.  Let $\existsBox \beta \not \in w$. Since $w$ is maximally consistent $\forallDiamond \neg \beta \in w$.

Consider any $i_{\existsBox \gamma} \in \hat{\I}$.  Now by lemma \ref{e-diamond consistency}(2), $\Gamma = \{\neg \beta, \gamma \} \cup \{ \psi \mid \forallBox \psi \in w\}$ is consistent. Thus there is some world $v \supseteq \Gamma$ and by construction of the canonical model, $w \xrightarrow{i_{\existsBox \gamma}}v$. Also since $\neg \beta \in v$ by induction $\hat{M}, v \models \neg \beta$. \\
 For $j \in \gamma(w)$, let $\top$ be any validity. By ($\existsBox$Nec) we have $\vdash_{\AxA} \existsBox \top$ and hence $\existsBox \top \in w$. Again, by lemma \ref{e-diamond consistency}(2), $\Gamma = \{\neg \beta, \top\} \cup \{ \psi \mid \forallBox \psi \in w\}$ is consistent. Hence there is some $v \supseteq \Gamma$. And thus $w\xrightarrow{j}v$ and by induction hypothesis, $\hat{M},v \models \neg \beta$.
 \\
Thus for every $a\in \gamma(w)$ there is some $v$ such that $w\xrightarrow{a}v$ and $\hat{\M},v \models \neg \beta$. Hence $\hat{\M},w \models \forallDiamond \neg \beta$.

For the other direction first note that $\existsBox \beta \in EB$. Now suppose $\existsBox \beta \in w$ then by definition of the canonical model we have for any $w\xrightarrow{i_{\existsBox \beta}}u$ it is always the case that $\beta \in u$. By induction hypothesis, for any $w\xrightarrow{i_{\existsBox \beta}}u$ we have $\hat{\M},u \models \beta$. Hence $\hat{\M},w \models \existsBox \beta$.
\end{proof}

\section{$\IQML$ bisimulation and elementary equivalence}
\label{sec: IQML bisimulation and elementary equivalence}

Modal logics are naturally associated with bisimulations. If two  pointed models
are bisimilar, the related worlds agree on propositions and satisfy the so-called
``back and forth'' property (\cite{blueBook}).  However, when we come to $\PTML$, 
since the agent set is not fixed, we need to have the notion of `world 
bisimilarity' as well as `agent bisimilarity'. Towards this, in \cite{PR17}, we 
introduce a notion of bisimulation for propositional term modal logic and 
show that it preserves $\PTML$ formulas. Similar definitions  of bisimulations 
for  first order modal logics can be found in \cite{vanBis,Wang17}.  

Now we introduce the notion of bisimulation for $\IQML$. Here the idea is that 
two worlds are bisimilar if they agree on all propositions and every index in 
one structure has a {\em corresponding index} in the other. The following definition 
of bisimulation formalizes the notion of `corresponding index'.

\begin{definition}
\label{def: IQML bisimulation}
Given two $\IQML$ models $\M_1$ and $\M_2$, an $\IQML$-bisimulation on them is a 
non-empty relation $G \subseteq (\W_1 \times \W_2)$  such that 
for all $(w_1,w_2)\in G$ the following conditions hold:
\begin{itemize}

\item[]
\begin{itemize}
\item[Val.] $\rho_1(w_1) = \rho_2(w_2)$.
\item[$\existsBox$forth.] For all $i\in \I_1$ there is some $j\in \I_2$ such that, for all $u_2$ such that $w_2\xrightarrow{j}u_2$, there is some $u_1$ such that $w_1\xrightarrow{i} u_1$ and $(u_1,u_2)\in G$.
\item[$\existsBox$back.] For all $j\in \I_2$ there is some $i\in \I_1$ such that, for all $u_1$ such that $w_1\xrightarrow{i} u_1$, there is some $u_2$ such that $w_2\xrightarrow{j}u_2$ and $(u_1,u_2)\in G$.
\item[$\existsDiamond$forth.] For all $i \in \I_1$ and for all $u_1$ such that $w_1\xrightarrow{i} u_1$, there is some $j\in \I_2$ and  some $u_2$ such that $w_2\xrightarrow{j}u_2$ and $(u_1,u_2)\in G$.
\item[$\existsDiamond$back.] For all $j \in \I_2$ and for all $u_2$ such that $w_2\xrightarrow{j}u_2$, there is some $i\in \I_1$ and some $u_1$ such that $w_1\xrightarrow{i} u_1$ and $(u_1,u_2)\in G$.
\end{itemize}
\end{itemize}
\end{definition}

Given two models $\M_1$ and $\M_2$ we say that $w_1, w_2$ are $\IQML$ bisimilar 
if there is some $\IQML$ bisimulation $G$ on the models such that 
$(w_1,w_2) \in G$ and denote it $(\M_1,w_1) \bisEq (\M_2,w_2)$. Also, we say 
$(\M_1,w_1) \elemEq (\M_2,w_2)$ if they agree on all $\IQML$ formulas i.e, 
for all $\phi \in \IQML$,\ $\M_1,w_1 \models \phi$ iff $\M_2,w_2 \models \phi$.

\begin{theorem}
\label{thm: bisimulation preserves formula equivalence}
For any two models $\M_1$ and $\M_2$ and any $w_1 \in \W_1$ and $w_2 \in \W_2$,\\ if $\M_1,w_1 \bisEq \M_2,w_2$ then $\M_1,w_1 \elemEq \M_2,w_2$.
\end{theorem}
\begin{proof}
Let $\M_1,w_1 \bisEq \M_2,w_2$ which means there is some $\IQML$ bisimulation $G$ such that $(w_1,w_2)\in G$. We need to show that for any $\phi \in \IQML$ we have $\M_1,w_1 \models \phi$ iff $\M_2,w_2 \models \phi$. 

We prove this for all $(v_1,v_2)\in G$ by induction on structure of $\phi$. The
base case and boolean cases are routine.

For the case $\phi := \existsBox \psi$:
Suppose $\M_1,v_1 \models \existsBox \psi$, we need to prove that $\M_2,v_2 \models 
\existsBox \psi$. Since $\M_1,v_1\models \existsBox \psi$, there is some $i\in \I_1$ 
such that $\M_1,v_1 \models \Box_i \psi$. Now let $j \in \I_2$ be the witness for 
$i$ for condition ($\existsBox$forth). We claim that $\M_2, v_2 \models \Box_j\psi$. 
Suppose not; then $\M_2,v_2 \models \Diamond_j \neg \psi$ and hence there is some 
$v_2 \xrightarrow{j} u_2$ such that $\M_2,u_2 \not\models \psi$. Since $j$ was 
the witness for $i$ for ($\existsBox$forth) condition, there is some 
$w_1 \xrightarrow{i}u_1$ such that $(u_1,u_2)\in G$. By induction hypothesis, 
$\M_1,u_1 \not\models \psi$ which contradicts $\M_1,u_1 \models \Box_i \psi$.
The other direction is proved symmetrically using ($\existsBox$back) condition.

\medskip
For the case $\existsDiamond \psi$: Suppose $\M_1,v_1 \models \existsDiamond \psi$ 
then there is some $i\in \I_1$ and some $u_1 \in W_1$ such that $v_1\xrightarrow{i}u_1$ 
and $\M_1,u_1\models \psi$. By condition ($\existsDiamond$forth) there is some 
$j\in \I_2$ and some $v_2\xrightarrow{j}u_2$ such that $(u_1,u_2)\in G$.  By 
induction hypothesis $M_2,u_2 \models \psi$ and hence $\M_2,v_2 \models 
\existsDiamond \psi$.  The other direction is symmetrically argued using 
$(\existsDiamond$back) condition.
\end{proof}

Now we prove that the converse holds over image finite models with finite index 
set $(\I)$.  $\M$ is said to be (index, image) finite if $\I$ is finite and 
$N^i(w)=\{u \mid (w,u) \in R_i\}$ is finite for all $w\in \W$ and $i\in \I$. 
 
 \begin{theorem}
 \label{thm:Formula equivalence corresponds to bisimulation over image finite models}
 Suppose $\M_1$ and $\M_2$ are (index,image) finite models then \\ $\M_1,w_1 \bisEq \M_2,w_2$ iff $\M_1,w_1 \elemEq \M_2,w_2$. 
 \end{theorem}
 \begin{proof}
 $(\Rightarrow)$ follows from Theorem \ref{thm: bisimulation preserves formula equivalence}. \\
 For $(\Leftarrow)$ suppose $\M_1,w_1 \elemEq \M_2,w_2$, then define  $G = \{ (v_1,v_2) \mid \M_1,v_1 \elemEq \M_2,v_2\}$. Note that $(w_1,w_2) \in G$. Hence it suffices to show that $G$ is indeed an $\IQML$ bisimulation. For this, choose any $(v_1,v_2)\in G$. Clearly $[Val]$ holds since $v_1,v_2$ agree on all $\IQML$ propositions. Now we verify the other conditions:
 
Now suppose that the ($\existsBox$forth) condition does not hold. Then there is some $\bold{i}\in \I_1$ such that for all $j\in \I_2$ there is some $u_j$(*) such that $v_2\xrightarrow{j} u_j$ and  for all $v_1 \xrightarrow{\bold{i}} u'$ we have $u' \not\elemEq u_j$. Let $\I_2 = \{ j_1\cdots j_{n}\}$ and let $u_l$ be the corresponding (*) for every $j_l$.
 Also let $\bold{i}$-successors of $v_1$ be $N_{\bold{i}}(v_1) = \{ s_1\cdots s_{m}\}$. By above argument,  we have $u_l\not\elemEq s_d$ for all $l\le n$ and $d\le m$. Hence for every $u_l$ and every $s_d \in N_{\bold{i}}(v_1)$ there is a formula $\phi^l_d$ such that $\M_1,s_d \models \phi^l_d$ but $\M_2,u_l \models \neg \phi^l_d$. Now consider the formula $\alpha = \existsBox ( \bigwedge\limits_{l}\bigvee\limits_{d} \phi^l_d)$. Note that  for all $l$ and for all $\bold{i}$-successors $s_d \in N_{\bold{i}}(v_1)$ we have $\M_1,s_d \models  \phi^l_d$ and hence  $\M_1, v_1 \models \Box_{\bold{i}} (\bigwedge\limits_{l} \bigvee\limits_{d } \phi^l_d)$ which implies $\M_1,w_d \models \alpha$. On the other hand for every $j_l \in \I_2$ at $u_l$ we have $\M_2 ,u_l \models \bigwedge\limits_d\neg \phi^l_d$ and hence $\M_2,v_2 \models \forallDiamond (\bigvee\limits_l \bigwedge\limits_d \neg \phi^l_d)$ which contradicts $v_1 \elemEq v_2$.
 \\ The ($\existsBox$back) condition is argued symmetrically.

\medskip
Suppose that the ($\existsDiamond$back) condition does not hold.  Then there 
is some $\bold{j} \in \I_2$ and some $w_2 \xrightarrow{j}\bold{u_2}$ such that 
for all $i \in \I_1$ and for all $w_1 \xrightarrow{i} u'$ we have 
$u' \not\elemEq \bold{u_2}$. Let $\R = \bigcup_{i\in \I_1}\R_i$ and let 
$N(w_1) = \{ u' \mid (v_1,u') \in \R\}$ be the set of all successors of $w_1$. 
Since $\M_1$ is (index, image) finite, let $N(w_1) = \{t_1\cdots t_{r}\}$. By 
above argument, for every $t_d \in N(w_1)$ there is a formula $\psi_d$ such that 
$\M_1,t_d \models \psi_d$ and $\M_2,\bold{u_2} \models \neg \psi_d$. Hence 
$\M_2,w_2 \models \Diamond_{\bold{j}}(\bigwedge\limits_d \neg \psi_d)$. Now 
consider $\beta = \existsDiamond (\bigwedge\limits_d \neg \psi_d)$. Clearly 
$\M_2,w_2 \models \beta$ (with $\bold{j}$ and $\bold{u_2}$ as witnesses). On 
the other hand, for any successor $t_d$ of $w_1$ since $\M_1,t_d \models \psi_d$ 
we have $\M_1,w_1 \models \forallBox (\bigvee\limits_d \psi_d)$ which contradicts
our assumption that $w_1$ and $w_2$ satisfy the same formulas.\\
The ($\existsDiamond$forth) is argued symmetrically.

 \end{proof}
 
An important consequence of the theorem above is that we can confine ourselves
to tree models for $\IQML$ formulas, since it is easily seen that an $\IQML$ 
model is bisimilar to its tree unravelling.

Given a tree model $\M$ we define its {\sf restriction} to level $n$ in the
obvious manner: $\M|n$ is simply the same as $\M$ upto level $n$ and 
the remaining nodes in $\M$ are `thrown away'.

We can now sharpen the result above: we can define a notion of $n$-bisimilarity
and show that it preserves $\IQML$ formulas with modal depth at most $n$.

\begin{definition}
\label{def: $n$-bisimulation}
Given two tree models $\M_1$ and $\M_2$, and $w_1$ in $M_1$,
$w_2$ in $M_2$, we say $w_1$ and $w_2$ are $0$-bisimilar if
$\rho_1(w_1) = \rho_2(w_2)$.

For $n > 0$, we say $w_1$ and $w_2$ are $n$-bisimilar if
the following conditions hold:
\begin{itemize}
\item[]
\begin{itemize}
\item[]
\begin{itemize}
\item[$n$-$\existsBox$forth.] For all $i\in \I_1$ there is some $j\in \I_2$ such that for all $w_2\xrightarrow{j}u_2$ there is some $w_1\xrightarrow{i} u_1$ such that $u_1$ and $u_2$
are $(n-1)$-bisimilar.
\item[$n$-$\existsBox$back.] For all $j\in \I_2$ there is some $i\in \I_1$ such that for all $w_1\xrightarrow{i} u_1$ there is some $w_2\xrightarrow{j}u_2$ such that $u_1$ and $u_2$
are $(n-1)$-bisimilar.
\item[$n$-$\existsDiamond$forth.] For all $i \in \I_1$ and for all $w_1\xrightarrow{i} u_1$ there is some $j\in \I_2$ and  some $w_2\xrightarrow{j}u_2$ such that $u_1$ and $u_2$
are $(n-1)$-bisimilar.
\item[$n$-$\existsDiamond$back.] For all $j \in \I_2$ and for all $w_2\xrightarrow{j}u_2$ there is some $i\in \I_1$ and some $w_1\xrightarrow{i} u_1$ such that $u_1$ and $u_2$
are $(n-1)$-bisimilar.
\end{itemize}
\end{itemize}
\end{itemize}
\end{definition}

We can now speak of an $n$-bisimulation relation between models and speak of
models being $n$-bisimilar, and employ the notation $(\M_1,w_1) \bisEq_n (\M_2,w_2)$.
Clearly, for tree models $(\M_1,w_1) \bisEq_n (\M_2,w_2)$ iff
$(M_1 | n, w_1) \bisEq (M_2 | n, w_2)$.

A routine re-working of the proof  of Theorem \ref{thm: bisimulation preserves formula equivalence}  shows that when two tree models
are $n$-bisimilar, they satisfy the same formulas of modal depth at most $n$.
That is, $(\M_1,w_1) \bisEq_n (\M_2,w_2)$ we have $(\M_1,w_1) \elemEq^n (\M_2,w_2)$.
We can go further and show that every $n$-bisimulation class is represented
by a single formula of modal depth at most $n$. For this, we assume (as is
customary in modal logic), that we have only finitely many atomic propositions.

\begin{lemma} Suppose that $\Ps$ is a finite set, then for any $n$ and for any $\M,w$ there is a formula $\chi^n_{[\M,w]}\in \IQML$ of modal depth $n$ such that for any 
$(\M',w') \models \chi^n_{[\M,w]}$ iff $(\M',w') \bisEq_n (\M,w)$.
\end{lemma}
\begin{proof}
Note that $(\Leftarrow)$ follows from Theorem \ref{thm: bisimulation preserves formula equivalence} specialized to $n$-bisimulation.
For the other direction, the proof is by induction on $n$. For $n = 0$, since $\Ps$ is finite, 
$\chi^0_{[\M,w]} = \bigwedge\limits_{p\in \rho(w)} p\land 
\bigwedge\limits_{q \not\in \rho(w)}\neg q$ is the required formula. 

Let $\R = \bigcup \R_i$ and let $\Gamma^n_\M = \{ \chi^n_{[\M,w]} \mid w \in \W\}$. Inductively $\Gamma^n_\M$ is finite. For any $S \subseteq \Gamma^n_\M$ let $\veeS$ denote the disjunction $\bigvee\limits_{\phi \in S}S$.  For the induction step,  the characteristic formula is given by: 
{\footnotesize
\begin{align*}
   \chi^{n+1}_{[\M,w]} &= \overbrace{\chi^0_{[\M,w]}}^{Val.} \land \overbrace{ \bigwedge_{i\in \I} \existsBox \big( \bigvee\limits_{(w,u)\in \R_i} \chi^n_{[\M,u]}\big)}^{n-\existsBox \text{forth}}\land  
\overbrace{ \bigwedge\limits_{S \subseteq \Gamma^n_\M}\big( \existsBox (\veeS) \implies \bigvee\limits_{i\in \I} \bigwedge\limits_{(w,u)\in \R_i}\forallBox (\chi^n_{[\M,u]} \implies \veeS)\big)}^{n-\existsBox \text{back}} \\
&\qquad{} \underbrace{\bigwedge_{(w,u)\in \R} \existsDiamond \chi^n_{[M,u]}}_{n-\existsDiamond\text{forth}}\ \land\  
\underbrace{\forallBox \big(\bigvee_{(w,u)\in \R}  \chi^n_{[M,u]}\big)}_{n-\existsDiamond\text{back}}
      \end{align*}
}

Note that the formula remains finite even if $\I$ is infinite or the number 
of successors of $w$ is infinite since inductively there are only finitely 
many characteristic formulas of depth $n$. We now prove that the formula $\chi^n_{\M,w]}$ indeed captures $n$-bisimulation.  First we verify that the formula $\chi^{n}_{[\M,w]}$ holds at $\M,w$:
     
     \begin{itemize}
     \item      $\M,w \models \chi^0_{[\M,w]}$ follows from the definition of $\rho$. 
     \item For the $n-\existsBox$ forth part, for every $i\in \I$ we have $\M,w \models \Box_i \big( \bigvee\limits_{(w,u)\in \R_i} \chi^n_{[\M,u]}\big)$ and hence the claim follows.
     \item For the $n-\existsBox$back part, let $S \subseteq \Gamma^n_{\M}$. Suppose $\M,w \models \existsBox \veeS$, let $\bold{j}$ be the witness. Hence we have $\M,w \models \Box_{\bold{j}}\veeS$. Now observe that for all $(w,u)\in \R_{\bold{j}}$ we have $\chi^n_{[\M,u]} \in S$, otherwise there is some $(w,u)\in R_{\bold{j}}$ such that $\M,u \models \bigwedge\limits_{\phi \in S} \neg \phi$ which  is a contradiction to $\M,w \models \Box_{\bold{j}}\veeS$. 
     
     Also, note that for any finite set of formulas $T$, if $\alpha \in T$ then $\alpha \implies (\bigvee\limits_{\psi \in T}\psi)$ is a propositional validity. 
     
     Now we need to show that $\M,w \models \bigvee\limits_{i\in \I} \bigwedge\limits_{(w,u)\in \R_i}\forallBox (\chi^n_{[\M,u]} \implies \veeS)$. For this, pick $i = \bold{j}$. 
     By above argument, for all $(w,u) \in \R_{\bold{j}}$ we have $(\chi^n_{[\M,u]} \implies \veeS)$ as a boolean validity. 
     Hence, we have $\M,w \models \bigwedge\limits_{(w,u)\in \R_{\bold{j}}} \forallBox (\chi^n_{[\M,u]} \implies \veeS)$.

     \item For $n-\existsDiamond$forth, let $(w,u) \in \R$ which means for some $i\in \I$ we have $(w,u)\in R_i$ and such that $\M,u \models \chi^n_{[\M,u]}$. Hence $\M,w \models \existsDiamond \chi^n_{[\M,u]}$.
     \item For $n-\existsDiamond$back, for any $i\in \I$ and any $(w,u)\in R_i$ we have $\M,u \models \chi^n_{[\M,u]}$ and hence $\M,w \models \forallBox (\bigvee\limits_{(w,u)\in \R} \chi^n_{[\M,u]})$.
     \end{itemize}
   \medskip  
     Now  suppose $\M',w' \models \chi^{n}_{[\M,w]}$, then we need to prove that $\M',w'$ is $(n)$-bisimilar to $\M,w$. We verify all the conditions:
     \begin{itemize}
     \item Condition $(Val)$ follows since $\M',w' \models \chi^0_{[\M,w]}$.
     \item For condition $(n-\existsBox$forth), let $i\in \I$. By $(\existsBox$forth) part $\chi^{n}_{\M,w]}$ we have\\ $\M',w' \models \existsBox \big( \bigvee\limits_{(w,u)\in \R_i} \chi^n_{[\M,u]}\big)$. Let $\bold{j'}\in \I'$ be the witness such that $\M',w' \models \Box_{\bold{j'}} \big( \bigvee\limits_{(w,u)\in \R_i} \chi^n_{[\M,u]}\big)$. Now for any $(w',u')\in \R'_{\bold{j'}}$ there is some $(w,u) \in R_i$ such that $\M',u' \models \chi^n_{[\M,u]}$ where $(w,u)\in \R_i$ and by induction hypothesis, $(u,u')$ are $n$-bisimilar.
     \item For condition $(n-\existsBox$back),  let $i' \in \I'$.\\
      Define $S = \{\chi^{n-1}_{[\M,u]} \mid$ for some  $(w',u')\in \R_{i'}$ we have $\M',u' \models \chi^{n-1}_{[\M,u]}\}$. Now clearly, $\M',w'\models \existsBox (\veeS)$. Hence by ($n$-$\existsBox$back) part of the formula, there is some $\bold{i}\in I$ for which $ \M',w' \models \bigwedge\limits_{(w,u)\in \R_i} \forallBox (\chi^n_{[\M,u]} \implies \veeS)$. 
      
      Now, let $T = \{\chi^{n-1}_{[\M,u]} \mid$ for some  $(w,u)\in \R_{\bold{i}}\}$. Note that $T \subseteq S$ (otherwise, there is some $(w,u)\in R_{\bold{i}}$ such that $\M,u \models \chi^{n-1}_{[\M,u]} \land \bigwedge\limits_{\phi \in S} \neg \phi$ which implies $\M,w \models \neg \bigwedge\limits_{(w,u)\in \R_i} \forallBox (\chi^n_{[\M,u]} \implies \veeS)$ and this is a contradiction to $\M,w \models \chi^n_{[\M,w]}$).      
       Hence, for every $(w,u)\in \R_{\bold{i}}$ there is some $(w',u')\in \R){i'}$ such that $\M',u' \models \chi^{n-1}_{\M,u]}$. Thus  $\bold{i}$ is the ($\existsBox$back) witness for $i'$.
       
       \item For the $n-\existsDiamond$forth condition, let $i\in \I$ and $(w,u)\in \R_i$. By $n-\existsDiamond$forth part of the formula, $\M',w' \models \existsDiamond \chi^{n-1}_{[\M,u]}$ and hence we have a corresponding $i'\in \I'$ and $(w',u')\in \R_{i'}$ such that $\M',u' \models \chi^{n-1}_{[\M,u]}$.
       
       \item Finally for $n-\existsDiamond$back, suppose $i'\in \I$ and $(w',u')\in \R_{i'}$ then by $n-\existsDiamond$back part of the formula, $\M',u' \models \chi^{n-1}_{[\M,u]}$ for some $i\in \I$ and $(w,u)\in R_i$. Thus we obtain the required witness.
     \end{itemize}

\end{proof}

\section{Bisimulation games and invariance theorem}
\label{sec: bisimulation games and invariance}

Like every propositional modal logic, $\IQML$ is also a fragment of first order
logic. However, implicit quantification over domain elements in $\IQML$ needs
to be made explicit as well as quantification over worlds. Since these serve
different purposes in the semantics, we use a two sorted first order logic.

\begin{definition}[$\SFO$ syntax]
\label{def: SFO syntax} Let $\VarX$ and $\VarTau$ be two countable and disjoint 
sorts of variables and $R$ a ternary predicate. The  two sorted $\FO$  $(\SFO)$, 
corresponding to $\IQML$  is given by:

$$ \alpha ::= Q_p(x) \mid R(x,\tau,y) \mid \neg \alpha \mid \alpha \land \alpha \mid \exists \tau\ \alpha \mid \exists x \ \alpha$$
where $Q_p$ is the corresponding monadic predicate for every $p\in \Ps$ and $x,y \in \VarX$ and $\tau \in \VarTau$.
\end{definition}

A $\SFO$ structure is given by $\Mfo = [(\W,\I),(\hat{R},\hat{\rho})]$ where $(\W,\I)$ is the two sorted domain and $(\hat{R},\hat{\rho})$ are interpretations with $\hat{R} \subseteq (\W \times \I \times \W)$ and $\hat{\rho}: \W \mapsto 2^{Q_\Ps}$ where $Q_\Ps = \{ Q_p \mid p\in \Ps\}$.
The semantics $\Vdash$ is defined for $\SFO$ in the standard way where the variables in $\VarX$ range over the first sort ($\W$) and variables of $\VarTau$ range over second ($\I$). 

Given an $\IQML$ structure $\M = (\W,\R_\I,\rho)$ the  corresponding $\SFO$ structure is given by $\Mfo = [(\W,\I),(\hat{R},\hat{\rho})]$  where $(w,i,v)\in \hat{R}$ iff $(w,v) \in \R_i$ and $Q_p \in \hat{\rho}(w)$ iff $p\in \rho(w)$. Similarly given any $\SFO$ structure, it can be interpreted as an $\IQML$ structure.  Thus there is a natural correspondence between $\IQML$ structures and $\SFO$ structures. For any $\IQML$ structure $\M$ let the corresponding $\SFO$ structure be denoted by $\Mfo$.

\begin{definition}[$\IQML$ to $\SFO$ translation]
\label{def: IQML to SFO translation}
The translation of $\phi \in \IQML$ into a $\SFO$ parametrized by $x\in \VarX$ is given by:

\begin{tabular}{l}
$\Tr( p\ :x) = Q_p(x)$ \\
$\Tr(\neg \phi\ : x) = \neg \Tr(\phi\ :x)$ \\
$\Tr(\phi \land \psi:\ x) = \Tr(\phi:\ x) \land \Tr(\psi\ :x)$\\
$\Tr(\existsBox \phi:\ x) = \exists \tau \forall y\ (R(x,\tau,y) \implies \Tr(\phi:\ y))$\\
$\Tr(\forallBox \phi:\ x)  = \forall \tau \forall y\ (R(x,\tau,y) \implies \Tr(\phi:\ y))$
\end{tabular}
\end{definition}

\begin{proposition}
\label{prop: bTML to SFO preservation}
For any formula $\phi \in \IQML$ and any $\IQML$ structure $\M$\\
$\M,w \models \phi$ iff $\Mfo,[x\mapsto w] \Vdash \Tr(\phi:\ x)$.
\end{proposition}

Hence $\IQML$ can be translated into $\SFO$ with $2$ variables of $\VarX$ sort and one variable of $\VarTau$ sort. Given two $\IQML$ models $\M_1$ and $\M_2$, the notion of $\IQML$ bisimulation  naturally translates to bisimulation over the corresponding $\SFO$ models $\Mfo_1$ and $\Mfo_2$.

Now we state the van Benthem type characterization theorem: bisimulation invariant
$\SFO$ formulas can be translated back into $\IQML$.  We say that $\alpha(x) \in \SFO$ is \emph{bisimulation invariant} if for all $\M_1,w_1 \bisEq \M_2,w_2$ we have $\Mfo_1,[x\mapsto w_1] \Vdash\alpha(x)$ iff $\Mfo_2,[x\mapsto w_2] \Vdash \alpha(x)$. We can similarly
speak of $\alpha(x)$ being $n$-bisimulation invariant as well.  Also, $\alpha(x)$ is equivalent to some $\IQML$ formula if there is some formula $\phi \in \IQML$ such that  for all $\M$ we have $\Mfo,[x\mapsto w] \Vdash \alpha(x)$ iff $\M,w \models \phi$.  

\begin{theorem}
\label{thm: vanBenthem like characterization}
Let $\alpha(x) \in \SFO$ with one free variable $x\in \VarX$. Then $\alpha(x)$ is bisimulation invariant iff $\alpha(x)$ is equivalent to some $\IQML$ formula.
\end{theorem}

Note that $ \Leftarrow $ follows from Theorem \ref{thm: bisimulation preserves formula equivalence}. 
To prove $(\Rightarrow)$ it suffices to show that if $\alpha(x)$ is bisimulation invariant then, for some $n$ it is $n$-bisimulation invariant, since we have already shown in the
last section that $n$-bisimulation classes are defined by $\IQML$ formulas.

Towards proving this, we introduce a notion of {\sf locality} for $\SFO$ formulas. For any tree model $\M$ and  let $\Mfo|n$ be the corresponding $\SFO$ model of $\M$ restricted to $n$ depth.

\begin{definition}
We say that a formula $\alpha(x)$ is $n$-local if for any tree model $(\M,w)$,
$\Mfo \Vdash \alpha(w)$ iff $\Mfo | n \Vdash \alpha(w)$.
\end{definition}

\begin{lemma}
\label{lemma: q. rank q implies 2^q local}
For any $\alpha(x) \in \SFO$  formula which is bisimulation invariant with $x \in \VarX$ then $\alpha(x)$ is $n$-local 
for $n=2^q$ where $q = q_x + q_\tau$ where $q_x$ is the quantifier rank of $\VarX$ 
sort in $\alpha(x)$ and $q_\tau$ is the quantifier rank of $\VarTau$ in $\alpha(x)$.
\end{lemma}

Assuming this lemma, consider a $\SFO$  formula $\alpha(x)$ which is bisimulation
invariant. It is $n$-local for a syntactically determined $n$. We now claim that
$\alpha(x)$ is $n$-bisimulation invariant. To prove this, consider
$\M_1,w_1 \bisEq_n \M_2,w_2$. We need to show that 
$\Mfo_1,[x\mapsto w_1] \Vdash\alpha(x)$ iff $\Mfo_2,[x\mapsto w_2] \Vdash \alpha(x)$.

Suppose that $\Mfo_1,[x\mapsto w_1] \Vdash\alpha(x)$. By locality,
$\Mfo_1 | n,[x\mapsto w_1] \Vdash\alpha(x)$. Now observe that
$\M_1 | n,w_1 \bisEq \M_2 | n,w_2$. By bisimulation invariance of
$\alpha(x)$, $\Mfo_2 | n,[x\mapsto w_2] \Vdash \alpha(x)$. But then
again by locality, $\Mfo_2,[x\mapsto w_2] \Vdash \alpha(x)$, and
we are done.

Thus it only remains to prove the locality lemma. For this, it is
convenient to consider the {\em Ehrenfeucht-Fraisse} ($\EF$)  game for
$\SFO$. In this game we have two types of pebbles, one for $\W$ and 
the other for $\I$.

The game is played between two players Spoiler($\PI$) and  Duplicator($\PII$) on two $\SFO$ structures. A configuration of the game is given by 
$[(\Mfo, \overline{s}); (\Mfo',\overline{t})]$ where $\overline{s} \in (\W\cup \I )^*$ is a finite string $(\W \cup \I)$ and similarly $\overline{t} \in (\W' \cup \I')^*$.

Suppose the current configuration is $[(\Mfo, \overline{s}); (\Mfo',\overline{t})]$.  In a $\W$ round, $\PI$ places a $\W$ pebble on some $\W$ sort in one of the structures and $\PII$ responds by placing a $\W$ pebble on a $\W$ sort in the other structure. In a $\I$ round, similarly  $\PI$ picks one structure and places an $\I$ pebble on some $\I$ sort and $\PII$ responds by placing an $\I$ pebble on some $\I$ sort in the other structure. In both cases, the new configuration is updated to $[(\Mfo, \overline{s}s); (\Mfo',\overline{t}t)]$ where $s$ and $t$ are the new elements(either $\W$ or $\I$ sort) picked in the corresponding structures.

A $(q_x,q_\tau)$ round game is one where $q_x$ many pebbles of type $\W$ 
are used and $q_\tau$ many pebbles of type $\I$ is used.  Player $\PII$ wins after $(q_x,q_\tau)$ if after $(q_x, q_\tau)$ rounds, if in $[(\Mfo, \overline{s}); (\M',\overline{t})]$ the mapping $f(s_i) = t_i$ forms a partial isomorphism over $\Mfo$ and $\Mfo'$. Otherwise $\PI$ wins.

It can be easily shown that $\PII$ has a winning strategy in the $(q_x,q_\tau)$ 
round game over two structures iff they agree on all formulas with quantifier 
rank of $\VarX$ sort $\le q_x$ and quantifier rank of $\VarTau$ sort $\le q_\tau$.

Let $\M,w$ be any tree structure. To prove lemma \ref{lemma: q. rank q implies 2^q local}, 
we need to prove that $\Mfo,w \models \alpha(x)$ iff $\Mfo | n \models \alpha(x)$.

Let $q = q_x + q_\tau$ and $\mathfrak{N}$ be $q$ disjoint copies of $\Mfo$ and $\Mfo |n$. Note that \emph{inclusion} relation $G$ over $\Mfo$ and $\Mfo| n$ forms a bisimulation. Also note that $G$ continues to be a bisimulation over the disjoint union of $\mathfrak{N}\uplus\Mfo, w$ and $\mathfrak{N}\uplus\Mfo |n,\ w$. Moreover, notice that $(\Mfo,w)$ is bisimilar to $(\mathfrak{N}\uplus\Mfo, w)$
 and further $(\Mfo |n,\ w)$ is bisimilar to $(\mathfrak{N}\uplus \Mfo|n,\ w)$.
 
 Now since $\alpha(x)$ is bisimulation invariant, it is enough to show that $\PII$ has a winning strategy in the game starting from $[(\mathfrak{N}\uplus\Mfo, w),\ (\mathfrak{N}\uplus\Mfo|n,\ w)]$.

The winning strategy for $\PII$ is to ensure that at every round $m < (q_x + q_\tau)$ the critical distance $d_m = 2^{q-m}$ is respected:

If $\PI$ places $\W$ pebble on a $\W$ sort which is  within $d_m$ 
of an already pebbled $\W$ pebble, $\PII$ plays according to a local isomorphism in the $d_m$-
neighbourhoods of previously pebbled elements (exists since $n = 2^q$ and $m < q$); if $\PI$ places a $\W$ pebble somewhere beyond $2^{q-m}$ distance from all $\W$ pebbles previously used, then, $\PII$ responds in a fresh isomorphic copy of
type $ \Mfo$ or $\Mfo | n$ correspondingly (again, it is guaranteed  to exist since previously at most $m-1 (< q)$ would have been used).

If $\PI$ decides to use an $\I$ pebble and places it  on some $\I$ sort $i$ in one structure, then  $\PII$ responds by placing an $\I$ pebble on $i$ in the mirror copy in the other structure, where by mirror copy we mean: for $\Mfo$ or $\Mfo|n$ in $\mathfrak{N}$ then the mirror copy in the other structure is itself and the original $\Mfo$ and $\Mfo|n$ are mirror copies of each other.

\section{Satisfiability problem}

The satisfiability problem for $\IQML$ can be solved by sharpening the
completeness proof of the axiom system by showing that every consistent
formula is satisfied in a model of bounded size. Indeed, a $\PSPACE$ 
decision procedure can be given along the lines of Grove and Halpern \cite{grove1}.
However, we give a tableau procedure for $\IQML$ which is instructive, and
as we will observe later, neatly generalizes to more expressive logics.

Given a formula $\phi$, we set $I = \{ c_\alpha \mid \existsDiamond \alpha \in SF(\phi)\} \cup \{ d_{\beta} \mid \existsBox \beta \in SF(\phi)\}$ where $SF(\phi)$, the set of
subfomulas of $\phi$ is defined in the standard way. This forms the index set where 
$c_\alpha$ and $d_\beta$ act as witnesses for the corresponding formulas.

We construct a tableau tree structure $T = (W,V,E,\lambda)$ where $W$ is a finite set, $(V,E)$ is a  rooted tree and $\lambda: V \mapsto L$ is a labelling map. Each element in $L$ is of 
the form $(w: \Gamma,i_\chi)$, where $w \in W$, $\Gamma$ is a finite 
set of formulas and $i_\chi\in I$.  The intended meaning of 
the label is that the node constitutes a world $w$ that satisfies the formulas in 
$\Gamma$ and $i_\chi$ is the incoming label edge of $w$.  

The tableau rules for $\IQML$ are inspired from the tableau procedure for the bundled fragment of first order modal logic introduced in \cite{PRW18}. The $(\land)$ and $(\lor)$ tableau rules are standard. For the modalities, the intuition  for the corresponding tableau rule is the following: Suppose that we are in an intermediate step of tableau construction when we have formulas $\{ \existsDiamond \alpha, \existsBox \beta, \forallDiamond \phi, \forallBox \psi\}$ to be satisfied at a node $w$. 
For this, first we need to add a new $c_\alpha$ successor node $wv_\alpha$ where $\alpha$ holds; this new node inherits not only $\alpha$ but also $\psi$.  Also, we need a $d_\beta$ successor which inherits $\beta, \phi$ and $\psi$. Finally for each $e_\gamma \in I$ we need a $\phi$-successor which  also inherits $\psi$.

The $(\BR)$ tableau rule extends this idea when there are multiple occurrences of each kind of formulas above.
In general if the set of formulas considered at node $w:(A,B,C,D)$ where  $A = \{ \existsDiamond \alpha_1 .. \existsDiamond\alpha_{n_1} \};\ B=  \{ \existsBox \beta_1 .. \existsBox \beta_{n_2}\};\ C = \{ \forallDiamond\phi_1 .. \forallDiamond \phi_{m_1}\}$ and $D =  \{ \forallBox \psi_1 .. \forallBox  \psi_{m_2}\}$. Let $D' = \{ \psi \mid \forallBox \psi \in D\}$. The $\BR$ rule is given as follows:

\begin{center}

\noindent \begin{tabular}{c}

$\dfrac{w:(A,B,C,D)}
{\{\langle wv_{\alpha_i}: (\alpha_i, D'),c_{\alpha_i} \rangle \mid i \le n_1\} \cup}
(\BR)$\\ \medskip
 $\{\langle wv^k_{\beta_j}: (\beta_j, \phi_k,D'), d_{\beta_j}\rangle \mid k \le m_1, j\le n_2\}\cup$\\ \medskip
$\{\langle wv^k_{e_\chi}: (\phi_k, D), e_\chi \rangle \mid l\le m_1, \chi \not\in (A \cup B)\}
$\\

\end{tabular}\\

\end{center}

From an `open tableau' we can construct a model for $\phi$, along the lines of \cite{PRW18}. Conversely it can be proved that every satisfiable formula has an open tableau. 

This tableau construction can be extended to the `bundled fragment' of full $\TML$ 
where we have predicates of arbitrary arity and the quantifiers and modalities occur 
(only) in the form $\forall x \Box_x \phi$ and $\exists x \Box \phi$ (and their 
duals). The proof follows the lines of \cite{PRW18}.

\section{Discussion}
We have studied the variable-free fragment of $\PTML$, with implicit modal
quantification. We could also consider more forms of implicit quantification
such as $\Box\forall$ and $\Diamond\forall$ modalities, though there is no
obvious semantics to them. These logics are the obvious variable free 
versions of monadic `bundled' fragments of $\TML$. One could consider a
similar exercise for `bundled' fragments of first order modal logic ($\FOML$).
As \cite{PRW18} shows, this is a decidable logic for increasing domain
semantics. 

Our study suggests that there are other forms of implicitly quantified 
modal logics. For instance, is there an implicit {\em hybrid} version 
of the logic studied by Wang and Seligman \cite{WangTML}?

A natural question is the delimitation of expressiveness of these logics: 
which are the properties of models expressed only by $\exists\Box$ or
only by $\forall\Box$ modalities ? How does nesting of these modalities
increase expressive power ? We believe that the model theory of implicit
modal quantification may offer interesting possibilities for abstract 
specifications of some infinite-state systems. However, for such study,
we will need to consider transitive closures of accessibility relations,
and this seems to be quite challenging.

Recent developments in tools for model checking and other decision 
procedures for fragments of $\FOML$ offer a promising direction to develop
similar practical frameworks for $\IQML$ and other decidable fragments of 
term-modal logics. Such tools can be of help in the synthesis and verification
of some classes of systems with unboundedly many agents.

\paragraph{Acknowledgement.} We thank Yanjing Wang for his insightful and extensive discussions on the theme of this paper.

\bibliographystyle{splncs04}
 \bibliography{ref}

\end{document}